\documentclass[aps,prl,twocolumn,superscriptaddress,reprint,amsmath,amssymb]{revtex4-2}

\usepackage{graphicx}% Include figure files
\usepackage{dcolumn}% Align table columns on decimal point
\usepackage{bm}% bold math
\usepackage{amsmath}
\usepackage{amssymb}
\usepackage{latexsym}
\usepackage{epsfig}
\usepackage{amsbsy}
\usepackage{array}
\usepackage{amssymb}
\usepackage{setspace}
\usepackage{mathtools}
\usepackage{amsthm}
\usepackage{physics}
\usepackage{algorithm}
\usepackage{algpseudocode}
\usepackage[colorlinks=true, urlcolor=blue, citecolor=black]{hyperref}
\usepackage{enumitem}
\usepackage{braket}
\usepackage{amsthm}
\usepackage{tikz}
\usepackage{overpic}
\usepackage{xcolor}
\usepackage{soul}

\newtheorem{theorem}{Theorem}

\newtheorem{proposition}[theorem]{Proposition}

\theoremstyle{definition}
\newtheorem{definition}[theorem]{Definition}
\usepackage{algpseudocode}

\usepackage[justification=centerlast]{caption}
\usepackage[labelformat=simple]{subcaption}

\usepackage[justification=centerlast]{caption}

\usepackage{xcolor}

\begin{document}

%\title{Complex networks with matrix weights} # v1
\title{Matrix-weighted networks for modeling multidimensional dynamics}

\author{Yu Tian}%
\affiliation{Nordita, Stockholm University and KTH
Royal Institute of Technology, Sweden}

\author{Sadamori Kojaku}
\affiliation{School of Systems Science and Industrial Engineering,
Binghamton University, USA}

\author{Hiroki Sayama}
\affiliation{School of Systems Science and Industrial Engineering,
Binghamton University, USA}
\affiliation{Waseda University, Japan}

\author{Renaud Lambiotte}
\affiliation{University of Oxford, UK}

% \date{}

\begin{abstract}
Networks are powerful tools for modeling interactions in complex systems.
While traditional networks use scalar edge weights, many real-world systems involve multidimensional interactions.
For example, in social networks, individuals often have multiple interconnected opinions that can affect different opinions of other individuals, which can be better characterized by matrices. 
%Matrix representations can effectively capture how one person's view on a topic influences not only the same topic but also other topics in their social connections.
We propose a novel, general framework for modeling such multidimensional interacting dynamics: matrix-weighted networks (MWNs).
We present the mathematical foundations of MWNs and examine consensus dynamics and random walks within this context.
Our results reveal that the coherence of MWNs gives rise to non-trivial steady states that generalize the notions of communities and structural balance in traditional networks.
%\sk{check: stationary state $\to$ steady state}
\end{abstract}

\maketitle
Networks have emerged as a powerful model for studying interacting systems across a wide range of disciplines, enhancing our understanding of complex systems such as the human brain, the Internet, and human societies \cite{newman2018networks}.
Networks are constructed from pairwise interactions between entities which, combined together, allow for signals to propagate in connected systems.
Much effort has been made to understand the intricate interplay between dynamical systems on networks and underlying network structure \cite{Jadbabaie2003coordination,Couzin2005animal,porter2016dynamics,lambiotte2022dynamics} with two main foci: examining the network properties that influence dynamics on networks, and utilizing dynamical processes to extract valuable information about underlying network structure.
For example, modular structure in networks can effectively simplify the description of dynamical systems \cite{lambiotte2022dynamics,rosvall2008maps} and,  in the presence of negative links, structural balance can help separate distinct dynamical behaviors \cite{cartwrightharary1956gbalance,Altafini2012opinion,tian2022sign}.

How to characterize pairwise interaction is a key modeling decision that significantly impacts the analysis of network dynamics.
Widespread choices include binary values indicating the presence or absence of an edge between two nodes, a positive real number to indicate connection strength, and a real value whose sign represents consolidation or contradiction~\cite{newman2018networks}. Recent studies have expanded the concept of edge weights to complex values~\cite{bottcher2022complex,tian2023complex} to represent the ``state'' of connections.

Recent years have also witnessed a movement towards multi-dimensional characterizations of node states and their interactions in various network models, such as multiplex networks and multi-layer networks \cite{kivela2014multilayer,sayama2016productG,sayama2018multilayer,bianconi2018multilayer}.
In general, higher-dimensional variables and multi-dimensional interactions naturally arise in a variety of applications \cite{schiwinger1960unitary,barooah2005electric,singer2012image,novikov2015NN,friedkin2016belief}.
For example, in image processing, each image necessarily lies in a higher-dimensional space, and it has been shown that more robust results can be obtained by considering the transformations in higher dimensions \cite{singer2012image}.
In social systems, it is natural that more than one statement is propagating through the networks, where the belief in one affects the belief in others \cite{friedkin2016belief,parsegov2016novel,ye2019consensus,nedic2019graph,ye2020continuous,ahn2020opinion,bizyaeva2023multi}.

However, the details of such nontrivial multidimensional interactions {\em between nodes} are not fully reflected in traditional network models. Edges in existing models with multidimensional node states are typically assumed to simply transmit a vector from one node to another, and/or to compare two vector states and react according to their distances, e.g.~\cite{friedkin2016belief}. These simple assumptions would be too limited to capturing more complex multidimensional interactions identified in various real-world network systems. This observation has led us to study multidimensional {\em interaction} models. 
Models for multidimensional interactions have been previously studied, particularly in opinion dynamics~\cite{axelrod1997dissemination,hansen2021opinion}. In these models, agents possess opinion vectors on various topics that can change in nontrivial ways through communication, depending on how agents interpret their neighbors’ opinions. While these models have found applications in machine learning~\cite{bodnar2022neural,braithwaite2024heterogeneous,barbero2022sheaf,gebhart2022graph}, their translation to physical systems and complex network dynamics remains largely unexplored, leaving opportunities to enhance our understanding of multifaceted interactions in complex networks.
%For instance in opinion dynamics, where agents possess a vector of opinions on different topics, opinion vectors can change in nontrivial ways through communication, depending on how the agents interpret their neighbors' opinions \cite{axelrod1997dissemination,hansen2021opinion}.
%There are also works in the machine learning community to incorporate multi-dimensional interactions into the architectures \cite{bodnar2022neural,braithwaite2024heterogeneous,barbero2022sheaf,gebhart2022graph}, but not yet in physics or network science in particular.
% Sheaf Graph Neural Networks (GNNs) offer a framework for studying multi-dimensional interactions on networks. Based on sheaf theory \cite{friedman2011sheaves,friedman2015sheaves1,friedman2015sheaves2}, these GNNs extend scalar edge weights to matrix-valued interactions \cite{bodnar2022neural,braithwaite2024heterogeneous,barbero2022sheaf,gebhart2022graph}.
% \todo{key difference: they focus on a matrix-weighted net for problem solving and did not use for modelling purpose.
% }
% \todo{Learn Sheaf theory}
% Despite the broad applicability and effectiveness, its theoretical understanding, in particular the dynamics on MWNs, remains limited.

%\sk{Why network approach? interpretability \& generalizability}
% \yu{[...More examples...]}

%\textit{Can we provide a network framework that naturally characterises this situation?}
In this article, we introduce a framework specifically designed for multidimensional interactions on networks, \textit{matrix-weighted networks} (MWNs). This is a generalization of conventional scalar- or complex-weighted network models, in which one can represent a linear transformation of state vectors between a pair of nodes connected by an edge. In most general settings, the weights can be any matrices of any size or shape (even rectangular ones may be allowed if the dimensions of node states differ from node to node) and self-loops are permitted, where not only multidimensional interactions of node states but also any network with any partitions can be analyzed from the angle of MWNs.
In this paper, we start with the cases of square matrix weights and directed edge reciprocity which, as we argue, can be seen as the generalization of classical undirected networks to the matrix-weighted setting.
More specifically, we consider the cases when edges characterize reciprocal interactions, where the effect of one characteristic dimension of node states on another dimension from one node to another is the same as the one when the orders of both nodes and dimensions are reversed.
%More specifically, we consider the cases when edges characterize transformations, where going from one node to another through a directed edge and then back to the original node through another directed edge will not change the position of the state vector in the higher dimensional space (ignoring the magnitude). Hence, the matrices that characterise the edges are orthogonal, up to proper scaling.
This choice lends itself to the study of orthogonal matrices, where isometries are defined on each edge.
We start by providing mathematical foundations for MWNs, and then
characterize the dynamical behavior through two key dynamics, consensus dynamics and random walks that are reformulated with this framework, both theoretically and numerically. Our main results concern the definition of the coherence of MWNs, which determines the existence of non-trivial steady states for the dynamics and is associated with large-scale patterns generalizing the notions of communities and structural balance. Finally, we discuss opportunities for future research and connections with existing network models. The code and documentation to reproduce all numerical results can be found at our GitHub repository~\cite{github}. 

% Definition of networks with matrix weights.
\paragraph{Matrix-weighted networks (MWNs).}
% Networks and matrix weights;
% (block) Laplacian matrix and normalized Laplacian matrix;
% (block) transition matrix.
\begin{figure}[htbp]
    \centering
    \includegraphics[width=.48\textwidth]{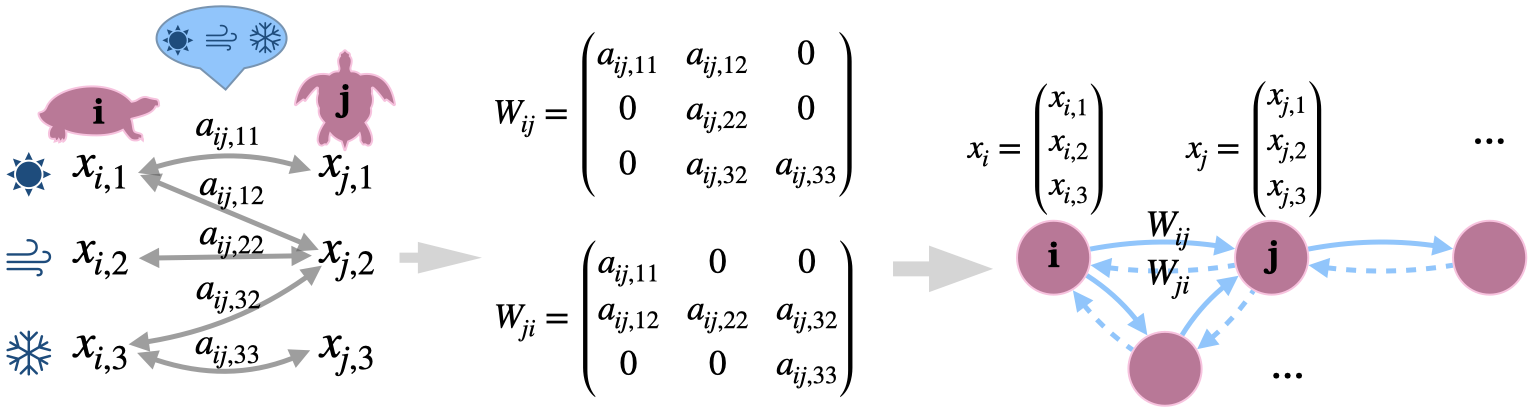}
    \caption{Illustration of a MWN, where each node is equipped with 3 state variables. In an opinion dynamics setting, each variable can be interpreted as an opinion on a specific topic. Here, we illustrate how the state vector of node $v_i$ affects the state vector of $v_j$ and vice versa. }
    \label{fig:eg}
\end{figure}
We consider networks in the form of weighted graphs $G = (V, E, \mathcal{W})$, being connected and directed, where $V = \{v_1, v_2, \dots, v_n\}$ is the node set and $E\subset V\times V$ is the edge set.
We use a matrix $\mathbf{W}_{ij}\in\mathbb{R}^{n_d\times n_d}$ to characterise the relationship between nodes $v_i$ and $v_j$, where $n_d$ is the characteristic dimension. $\mathbf{W}_{ij}$ is the all-zero matrix if and only if there is no directed edge from $v_i$ to $v_j$. Otherwise, the matrix encodes how a $n_d$-dimensional signal on node $v_j$ impacts the signal on node $v_i$. See Fig.~\ref{fig:eg} for an example.
We define the $(n~n_d) \times (n~n_d)$ supra-weight matrix $\mathcal{W}$ to have block structure, with the $i,j$ block being $\mathbf{W}_{ij}$, the matrix weight between nodes $v_i$ and $v_j$.
Furthermore, we decompose each matrix weight into two parts,
% one accounting for the transformation $\mathbf{R}_{ij}\in O(n_d)$ being an orthogonal matrix and the other for the magnitude $w_{ij} = \norm{\mathbf{W}_{ij}}_2$,
one accounting for the magnitude $w_{ij} = \norm{\mathbf{W}_{ij}}_2$, and the other for the transformation $\mathbf{R}_{ij}$ with $\norm{\mathbf{R}_{ij}}_2 = 1$ (for $w_{ij} > 0$),
such that $\mathbf{W}_{ij} = w_{ij}\mathbf{R}_{ij}$.
In this paper, we assume that $\mathbf{W}_{ij} = \mathbf{W}_{ji}^T$,
i.e., the interaction between the same pair of dimensions is reciprocal between the same pair of nodes.
Then we have $w_{ij} = w_{ji}$, $\mathbf{R}_{ij} = \mathbf{R}_{ji}^T$, and finally, $\mathcal{W}^T = \mathcal{W}$. This choice ensures that the eigenvalues of $\mathcal{W}$ are real numbers, and can be understood as an extension of undirected networks to the matrix-weighted setting.

We define the degree (or strength) of a node $v_i$ as the sum of the matrix norm of edges incident on this node,
\begin{align*}
    d_i = \sum_j\norm{\mathbf{W}_{ij}}_2 = \sum_{j}w_{ij},
\end{align*}
and the supra-degree matrix as
\begin{align*}
    \mathcal{D} = \mathbf{D}\otimes \mathbf{I},
\end{align*}
where $\mathbf{D}$ is the diagonal matrix with $\mathbf{d} = (d_i)$ on its diagonal, $\mathbf{I}$ is the identity matrix in $\mathbb{R}^{n_d\times n_d},$ and $\otimes$ denotes Kronecker product.
We define the supra-Laplacian matrix as
\begin{align*}
    \mathcal{L}
    &= \mathcal{D} - \mathcal{W}.
\end{align*}
We note that there are other types of Laplacian considered in the multidimensional setting, such as the one based on node-edge interactions \cite{hansen2021opinion}, but ours is more consistent with existing network science models, thus more appropriate to adapt network-based methods.
If we consider a vector $\mathbf{x} = (\mathbf{x}_1; \mathbf{x}_2; \dots; \mathbf{x}_n)$ with $\mathbf{x}_i\in\mathbb{R}^{n_d}$ for each $i$, then we can write
\begin{align}
    \mathbf{x}^T\mathcal{L}\mathbf{x} = \sum_{(v_i,v_j), (v_j,v_i)\in E}w_{ij}\norm{\mathbf{x}_i - \mathbf{R}_{ij}\mathbf{x}_j}_2^2 \ge 0.
    \label{equ:laplacian-bilinear}
\end{align}
Hence, the supra-Laplacian matrix maintains the property of the matrix Laplacian to be positive semi-definite.
We then define the supra-random-walk Laplacian matrix as
\begin{align*}
    \mathcal{L}_{rw}
    &= \mathcal{I} - \mathcal{P},
\end{align*}
where $\mathcal{I}\in\mathbb{R}^{nn_d\times nn_d}$ is the identity matrix and $\mathcal{P} = \mathcal{D}^{-1}\mathcal{W}$ is the supra-transition matrix.
Clearly, if $n_d = 1$, we retrieve the classic definitions. We will show later that the supra-Laplacian matrix and the supra-random-walk Laplacian matrix are closely related to the consensus dynamics and the random walks, respectively, in the multi-dimensional case.

% \begin{proposition}
%     The \new{supra-Laplacian} $\mathcal{L}$ is positive semi-definite.
%     \label{pro:L-semi-definite}
% \end{proposition}
% \begin{proof}[Proof sketch]
%     This can be proved by the bilinear form of $\mathcal{L}$, where for all $\mathbf{x} = (\mathbf{x}_1; \mathbf{x}_2; \dots; \mathbf{x}_n)$ with $\mathbf{x}_i\in\mathbb{R}^{n_d}$ for each $i$,
%     \begin{align}
%         \mathbf{x}^T\mathcal{L}\mathbf{x} = \sum_{(i,j), (j,i)\in E}w_{ij}\norm{\mathbf{x}_i - \mathbf{R}_{ij}\mathbf{x}_j}_2^2 \ge 0.
%         \label{equ:laplacian-bilinear}
%     \end{align}
% \end{proof}

% Structural properties: the notion of balance.
\paragraph{Coherence.} We define the \textit{transformation} of a directed path
% \footnote{\footnotesize{Note that the term ``directed paths (cycles)" are generally used in directed graphs, but since we assume that $\mathcal{W}$ is symmetric, the existence of a path (cycle) without direction is equivalent to the existence of a path (cycle) with either direction in $G$. Hence, we directly use ``paths (cycles)" in this paper, and specify the direction when necessary.}}
$P$ with composing edges $e_1, e_2 \ldots, e_k$ as
\begin{align*}
    \mathbf{R}(P) = \mathbf{R}(e_1)\mathbf{R}(e_2)\dots\mathbf{R}(e_k),
\end{align*}
where $\mathbf{R}(e)$ returns the transformation of edge $e$. The transformation of a directed cycle is defined similarly.
This notion is essential to understand how signals propagate along walks in the graph. As we will see, in certain graphs, different walks to the same node will be associated to different linear transformations, and hence to conflicting signals, while in other graphs, walks to the same node will always result in the same linear transformation and thus to reinforcement akin to resonance. The presence or absence of conflicts along the walks is described by the notion of {\em coherence}, which we define now, and determines the asymptotic of the dynamics studied in this paper.

% \begin{definition}
%     Let $G = (V, E, \mathcal{W})$ be a matrix-weighted graph.
%     $G$ is \textit{coherent} if and only if the transformation of every directed cycle is $\mathbf{I}$.
%     \label{def:balance}
% \end{definition}
We define a matrix-weighted network $G$ to be \textit{coherent} if and only if the transformation of every directed cycle is $\mathbf{I}$.
An immediate consequence for a graph to be coherent is that all matrix weights are orthogonal under the assumptions adopted in this paper, since for each edge $(v_i, v_j)\in E$, we also have $(v_j, v_i)\in E$, and for the directed cycle of the two reciprocal edges to have transformation $\mathbf{I}$, we have $\mathbf{R}_{ij}\mathbf{R}_{ji} = \mathbf{R}_{ij}\mathbf{R}_{ij}^T = \mathbf{I}$. Note that this definition, which naturally describes the geometrical problems considered in this paper, is strict, hence we also quantify the \textit{level of coherence} by how far the network is from being coherent, e.g., through the portion of cycles that do not have $\mathbf{I}$ as the transformation.
Furthermore, the notion can also be relaxed in other contexts, as we discuss in Supplemental Material (SM).

\begin{figure}[htbp]
    \centering
    \includegraphics[width=0.45\linewidth]{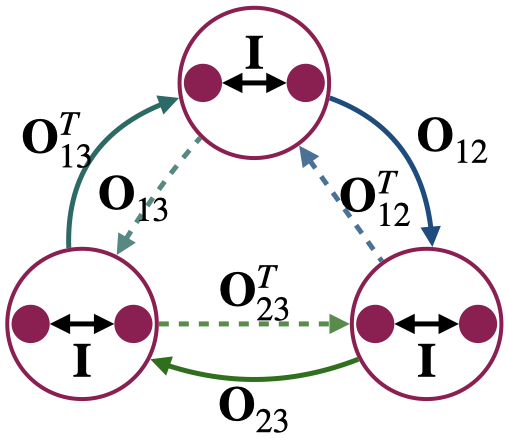}
    \caption{Example of the block structure of a coherent MWN, where the label indicates the orthogonal matrix weights and $\mathbf{O}_{13} = \mathbf{O}_{12}\mathbf{O}_{23}$}
    \label{fig:blocks}
\end{figure}
The notion of coherence also characterizes particular block structures in the MWN. Suppose there is a partition of $G$ such that (i) any edges within each block have transformation $\mathbf{I}$, and (ii) any edges between the same pair of blocks have the same transformation, thus we can consider each block as a super node; see Fig.~\ref{fig:blocks}. Then if (iii) any directed cycles of such super nodes have transformation $\mathbf{I}$, we can immediately see that the MWN is coherent, since any directed cycles in the network, apart from edges inside each block, form directed cycles of the super nodes (if not empty). We also show that such block structure always exists in coherent MWN in SM, hence uniquely defining the coherence.

With the characteristic block structure, we can connect all coherent networks with the \textit{identity-transformed MWN} where the transformation of each edge is $\mathbf{I}$. The identity-transformed MWN is a typical coherent MWN, and has supra-Laplacian
\begin{align*}
    \bar{\mathcal{L}} = \bar{\mathbf{L}} \otimes \mathbf{I},
\end{align*}
where $\bar{\mathbf{L}} = \mathbf{D} - \bar{\mathbf{W}}$ and $\bar{\mathbf{W}} = (w_{ij})$ only maintains the magnitude of each edge weight, thus $\bar{\mathbf{L}}$ is the classic graph Laplacian for positive scalar weights.
For a coherent MWN, we encode its intrinsic block structure into a block diagonal matrix, $\mathcal{S}$: its $i$-th block on the diagonal returns the transformation of any directed paths from nodes in the first block to the block that node $v_i$ belongs to. Then,
\begin{align*}
    \mathcal{L} = \mathcal{S}^T\bar{\mathcal{L}}\mathcal{S}.
\end{align*}
Hence, the supra-Laplacian also has eigenvalue $0$, and the associated eigenvectors span the subspace
\begin{align*}
    \mathbf{X} = \mathcal{S}^T(\mathbf{1}\otimes \mathbf{I}).
\end{align*}
We also show that the existence of $0$ as an eigenvalue always leads to coherence in SM, hence also intrinsic to the coherent MWN.  

\paragraph{Consensus dynamics.}
Consensus has been one of the most important dynamics on networks, with various applications \cite{lynch1996algorithm,balch2002robot,russel2002satellite,fax2004vehicle}.
Note that matrix relationship (or \textit{logic constraints}) has been considered in the belief formation of multiple statements \cite{friedkin2016belief}, but we incorporate transformation in the propagation process, as we will see later.
Now, let us endow each node with a vector state variable, $\mathbf{y}_i\in\mathbb{R}^{n_d}$, and the interdependence between the dimensions are encoded in the matrix weights $\{\mathbf{W}_{ij}\}$.
The \textit{consensus dynamics} can then be defined as
\begin{align*}
    \Dot{\mathbf{y}}_i = \sum_{j}w_{ij}(\mathbf{R}_{ij}\mathbf{y}_j - \mathbf{y}_i),
\end{align*}
i.e., each node adjusts its state such that the difference to its neighbors, after the characteristic transformation, is reduced, with information from all dimensions integrated.
Summarising the states of all nodes into a vector $\mathbf{y}= (\mathbf{y}_1; \dots; \mathbf{y}_n)$, we have $\Dot{\mathbf{y}} = -\mathcal{L}\mathbf{y}$.

The node state vectors then exhibit distinct behaviors in the multidimensional space, depending on the underlying MWN.
If the network is coherent, multi-consensus can be achieved at the steady state, where nodes in different blocks are characterized by different vectors (see SM for more detail).
In all other networks, the consensus dynamics will reach the steady state of all-zero vector for all nodes, i.e., global consensus, in a sufficiently long time, with
the relaxation time depending on the level of coherence.
We show later that the consensus dynamics on such networks can exhibit time-scale separation, where the states of nodes approach multi-consensus in a short time while converging to the all-zero state in the long run.

\paragraph{Random walks.} Random walks is another important dynamics on networks, aiming at modeling how an entity randomly explores the underlying structure, with applications in various fields \cite{masuda2017RW,fisher1966rwchem,pinkski1976pagerank,page1999pagerank,okubo2001diffusion,ewens2004mathematical,gold2007neuralDM,gyllingberg2024minimal}.
Now, let us consider a random walker that can not only randomly move on the network but change its states in its characteristic space $\mathbb{R}^{n_d}$. Specifically, for a random walker on node $v_i$ at time $t$, it will choose one of $v_i$'s neighbors randomly at $t+1$, with probability $w_{ij}/d_i$ for each neighbor $v_j$, while the characteristic vector will also transform according to the edge connecting them, through $\mathbf{R}_{ij}$.
Hence, if we average the characteristic vectors at each node by the corresponding probability, denoted by $\mathbf{y}_i$ for each node $v_i$, the (discrete-time) \textit{random walks} on the network can be described as
\begin{align*}
    \mathbf{y}_j(t+1)
    = \sum_{i}\frac{w_{ij}}{d_i}\mathbf{R}_{ij}^T\mathbf{y}_i(t)
    = \sum_{i}\frac{\mathbf{W}_{ij}^T\mathbf{y}_i(t)}{d_i}.
\end{align*}
Furthermore, if we consider $\mathbf{y} = (\mathbf{y}_1; \dots; \mathbf{y}_n)$, then $\mathbf{y}(t+1) = \mathcal{P}^T\mathbf{y}(t) = (\mathcal{P}^T)^{t+1}\mathbf{y}(0)$, where $\mathcal{P}$ is the supra-transition matrix.

With multidimensional interaction incorporated, the random walkers can explore the multidimensional sphere and exhibit averaged behavior that varies in accordance with the structure of MWN. If the network is coherent (and not bipartite), the averaged vector at the steady state will have different directions for nodes in different blocks (see SM for more detail).
In all other networks, either there will be no steady state, or the averaged vector of each node will be all zero at steady state, i.e., there is no particular direction that is preferred by the random walkers. The relaxation time depends on the level of coherence.

\paragraph{Numerical results.}
We validate our theoretical results by using a matrix-weighted stochastic block model (MSBM).
We generate a base network using a standard stochastic block model with $n_b$ blocks, edge probability $p_{\text{in}}$ within communities and $p_{\text{out}}$ between communities ($p_{\text{in}} \geq p_{\text{out}}$).
We first create a random base rotation matrix $\mathbf{R} \in \mathbb{R}^{n_d \times n_d}$~\footnote{
We apply a series of random 2D rotations, for any pair of dimensions in a random order, to the identity matrix, with the angle drawn from a uniform distribution over $[-\pi, \pi]$. 
}.
For nodes $v_i$ and $v_j$ in blocks $k$ and $\ell$ respectively, we then set $\mathbf{R}_{ij} = \mathbf{R}^{\ell - k}$, ensuring coherence.
To introduce variability, we perturb each $\mathbf{R}_{ij}$ based on a disturbance parameter $\eta \in [0,1]$ and a noise intensity $\sigma \geq 0$. With probability $\eta$, we apply additional random rotations with angles drawn from $\mathcal{N}(0, \sigma^2)$.
We use 120 nodes divided into equal-sized communities. We test various configurations with dimensions $n_d \in \{2, 3, 5, 10\}$, edge probabilities $p_{\text{in}} = 0.3$ and $p_{\text{out}} \in \{0.1, 0.3\}$, number of communities $n_b \in \{2, 3, 4\}$, perturbation probabilities $\eta \in \{0, 0.2\}$, and noise intensities $\sigma^2 \in \{0, 0.1, 0.3\}$.
All the configurations produced consistent results, and we report the representative results (see also SM).

\begin{figure*}[htbp]
    \centering
    \includegraphics[width=\linewidth]{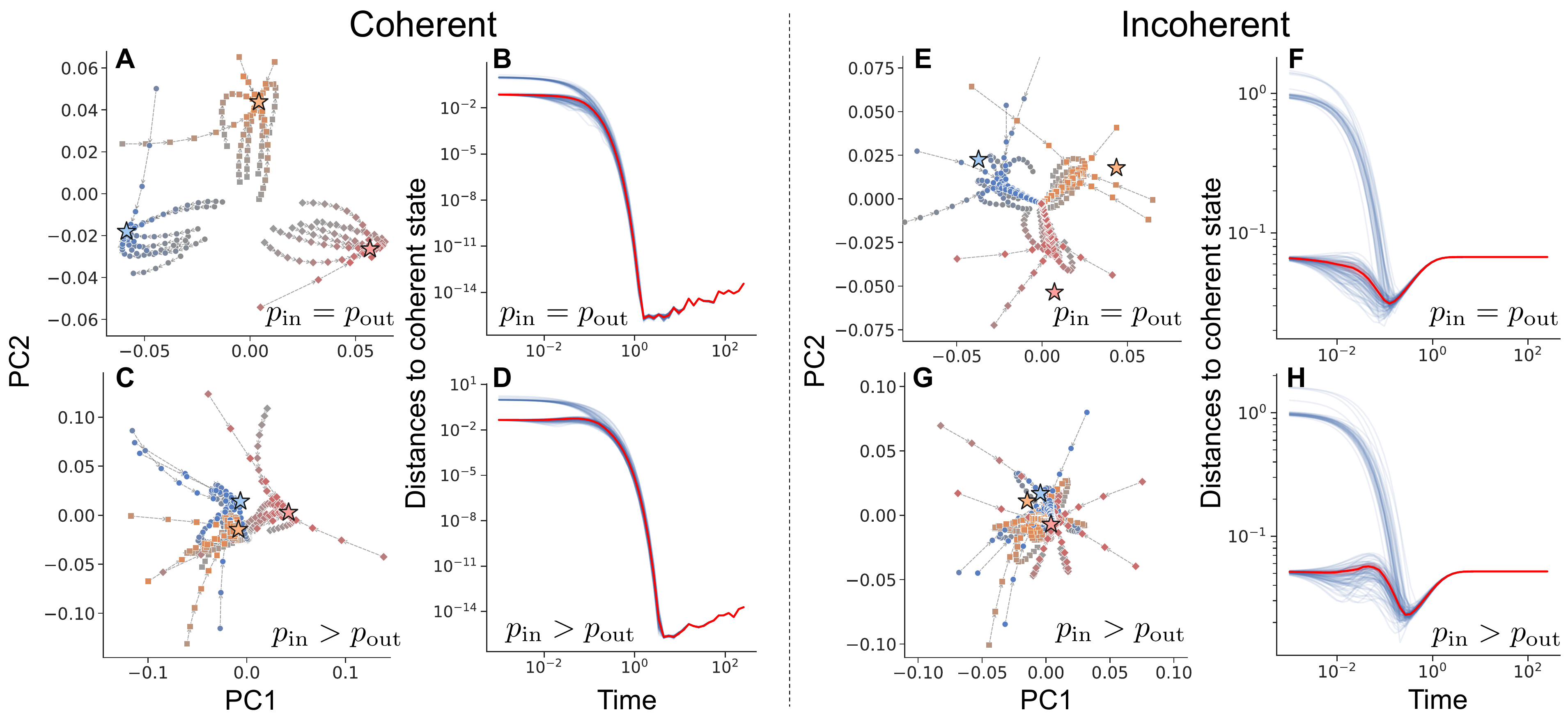}
    \caption{Consensus dynamics in a three-block MSBM with $120$ nodes and $10$ dimensional state space.
    Red, yellow, and blue represent nodes in blocks 1, 2, and 3, respectively.
    Stars indicate theoretical steady states.
    We set $n_d = 10$, $p_{\text{in}} = 0.3$, $p_{\text{out}} = \{0.1, 0.3\}$.
    {\bf A, B}: Coherent case, no topological community structure.
    {\bf C, D}: Coherent case, strong topological community structure.
    {\bf E, F}: Incoherent case, no topological community structure.
    {\bf G, H}: Incoherent case, strong topological community structure.
    Panels A, C, E, G show PCA projections of state vector trajectories.
    In panels B, D, F, H, each blue line represents the distance to the steady state of a node, and the red line represents the average distance.
    The colorband indicates the 95\% confidence interval estimated by $10^3$ bootstrap sampling.
    }
    \label{fig:sim-results-consensus}
\end{figure*}

The consensus dynamics on coherent MWNs align closely with our theoretical predictions (Fig.~\ref{fig:sim-results-consensus}A--D).
The node states converge towards their theoretical steady states, followed by a slow deviation within the distance of $10^{-13}$.
The incoherent case shows a distinct behavior: the state vectors first approach the steady state in the corresponding coherent case in a shorter time scale, followed by a notable deviation from the theoretical steady states (Figs.~\ref{fig:sim-results-consensus}E--H).
We observed the same behavior for different dimensions $n_d$ and different numbers of communities $n_b$ (see Robustness Analysis section in SM).
Notably, these dynamics for both coherent and incoherent cases remain consistent regardless of the presence ($p_{\text{in}}> p_{\text{out}}$) or absence ($p_{\text{in}}= p_{\text{out}}$) of topological community structure in the base network.

\begin{figure*}[htbp]
    \centering
    \includegraphics[width=\linewidth]{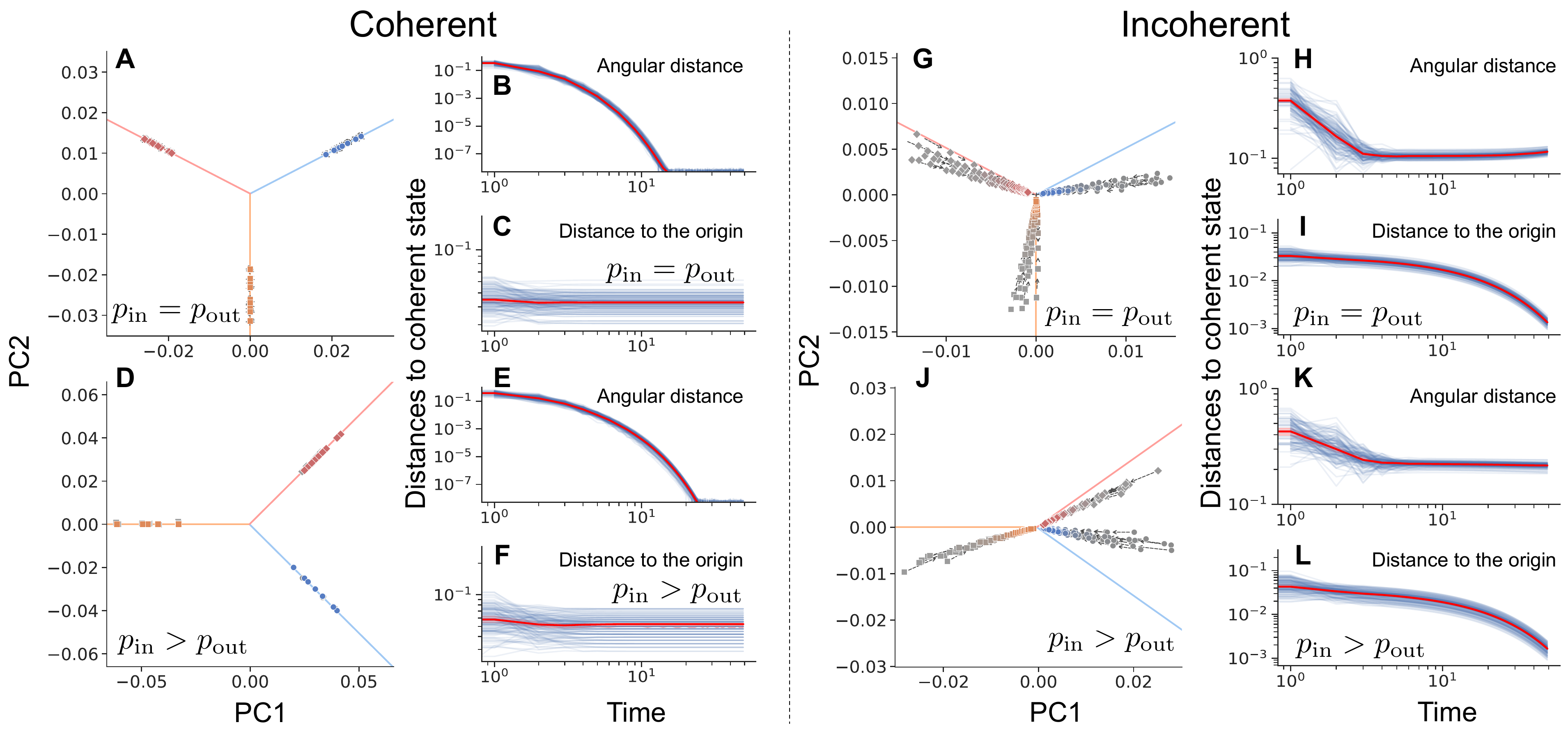}
    \caption{Random walk dynamics in a three-block MSBM with $120$ nodes and $10$ dimensional state space.
    Red, yellow, and blue represent nodes in blocks 1, 2, and 3, respectively.
    The directions of the theoretical steady states are indicated by solid lines.
    We set $n_d = 10$, $p_{\text{in}} = 0.3$, $p_{\text{out}} = \{0.1, 0.3\}$.
    {\bf A--C}: Coherent case, no topological community structure.
    {\bf D--F}: Coherent case, strong topological community structure.
    {\bf G--I}: Incoherent case, no topological community structure.
    {\bf J--L}: Incoherent case, strong topological community structure.
    Panels A, D, G, J show PCA projections of state vector trajectories.
    Panels B, E, H, K show angular distances to steady state (blue: individual nodes, red: average).
    Panels C, F, I, L show distances to origin (blue: individual nodes, red: average).
    The colorband indicates the 95\% confidence interval estimated by $10^3$ bootstrap sampling.
    }
    \label{fig:sim-results-random-walks}
\end{figure*}

The random walk dynamics on MWNs align with our theoretical predictions (Fig.~\ref{fig:sim-results-random-walks}). We compare node states to their predicted steady states using angular distance rather than Euclidean distance, since they have the same direction instead of the exact states. We note that theoretical predictions focus on average behavior and are not normalized to the unit sphere, while individual node states in the coherent case are on the unit sphere.
In coherent networks, node states converge angularly to their predicted steady states (Figs.~\ref{fig:sim-results-random-walks}A, B, D, and E).
For incoherent networks, we observe a two-phase behavior. Initially, node states converge angularly to the steady state, followed by a stable deviation (Figs.~\ref{fig:sim-results-random-walks}G, H, J, K). The node states are drawn towards the origin (Figs.~\ref{fig:sim-results-random-walks}G, I, J, L), which occurs because node states are randomly mixed and averaged, smoothing out differences and resulting in a stable state at the origin.

\paragraph{Discussions.}
In this paper, we introduced a novel framework of MWNs, laying the mathematical groundwork for analyzing both their structural and dynamical properties. This approach opens new avenues for research in network science. For instance, random walks underpin numerous techniques for extracting insights from network data, such as determining node importance through centrality measures \cite{gleich2015pagerank} and developing efficient community detection algorithms \cite{rosvall2008maps,lambiotte2014random}. Our framework provides an opportunity to extend these methods to datasets characterized by high-dimensional, interacting features. Our work also opens several fundamental questions for future research. In particular, we have explored in detail a strict notion of coherence in this paper but, as we discuss in SM, relaxed versions of coherence could be considered, leading to their own challenges.

This work participates in the recent efforts to enrich the standard network paradigm to ``higher orders" \cite{lambiotte2019networks}.
MWNs could provide a new angle to analyze and generalize models of multilayer networks \cite{kivela2014multilayer}  and temporal networks \cite{holme2012temporal},
%but also characterize any networks with any partitions in its most general form,
emphasizing the interaction between different dimensions (or layers).
Take a system where each node is copied in $n_d$ layers. Its connections via interlayer and intralayer connections to other nodes can be encoded in a matrix, making it possible to represent the system as a MWN. In addition, our definition of the matrix of a path, as the product of matrices of each of its edges, is a natural generalization of switching networks \cite{masuda2013temporal}, a popular model for temporal networks, where the system is modeled as a random sequence of matrices, instead of a sequence determined by a walk on a graph as we consider here.
In fact, any network with any node groupings can be modeled as an MWN, by defining the state of a node group as a vector and the interaction between groups as a matrix.
%\new{MWNs can also be applied to analyze the dynamics of groups of nodes, e.g., obtained from a community detection method, where each group would correspond to a meta-node and its state would be a vector .}
We hope that MWNs not only deepen our understanding of complex systems but also provide a flexible framework for modeling a wide range of systems that have multidimensional dynamics at multiple scales.

%For example, in a multilayer network, corresponding to multiple relationships under consideration, each node is usually copied $n_l$ times to construct $n_l$ layers: each layer can correspond to one relationship, and the interlayer edges encode the interaction between different relationships. However, our framework emphasizes the nontrivial interaction between different relationships, or characteristic dimensions, which is usually simplified in existing models, and consider a matrix as edge weight.
% If we represent a multilayer network via a supra weight matrix, denoted by $\mathcal{W}' = (\mathbf{W}^{k,l}) \in \mathbb{R}^{nn_l\times nn_l}$, where $\mathbf{W}^{k,l}\in\mathbb{R}^{n\times n}$ encodes the (scalar) edge weights between nodes in layers $k$ and $l$.
% Now let us consider the $i,j$ element of $\mathbf{W}^{k,l}$ for all possible layers $k,l$, denoted by $\mathbf{W}_{ij}^{:,:}\in \mathbb{R}^{n_l\times n_l}$, it characterises the relationship between each pair of layers, or dimensions, between nodes $v_i,v_j$, and is a weight matrix in our framework. Hence, $\mathcal{W} = (\mathbf{W}_{ij}^{:,:})$ retrieves the \new{supra-weight} matrix in our framework.
% An usual assumption in multilayer networks is that $\mathbf{W}^{k,l} = \mathbf{I}$ if $k\ne l$, and in temporal networks, $\mathbf{W}^{k,l} = \mathbf{0}$, all-zero matrix, if $\abs{k-l} > 1$.

% \yu{[Relationship with higher-order Markov chain.]}

\begin{acknowledgments}{\bf Acknowledgments}
R.L. acknowledges support from the EPSRC grants EP/V013068/1, EP/V03474X/1 and EP/Y028872/1.
Y.T.~is funded by the Wallenberg Initiative on Networks and Quantum Information (WINQ).
\end{acknowledgments}

\bibliography{refs}

\appendix

% \section{SUPPLEMENTARY INFORMATION}
\section{SUPPLEMENTAL MATERIAL}
\subsection{Relaxed notion of coherence}
At the core of the notion of coherence for matrix-weighted networks (MWNs), there is the condition that the transformation associated with any directed cycle $C$, e.g., starting at node $v_i$ and composed of the edges $e_1, e_2 \ldots, e_k$,
\begin{align*}
    \mathbf{R}(C) = \mathbf{R}(e_1)\mathbf{R}(e_2)\dots\mathbf{R}(e_k),
\end{align*}
is the identity matrix.
Intuitively, the idea is that a signal, starting at node $v_i$, after it propagates along any cycle, will be mapped onto the same signal, solely from the transformations. This makes the existence of a non-trivial steady state possible.
% As we have discussed, in the case considered in this paper where $\mathbf{W}_{ij} = \mathbf{W}_{ji}^T$, a necessary condition for coherence is that the matrix weight of each edge is orthogonal.
However, there could be situations where this condition is too strict.
% For instance, in the situations when the coupling matrix is symmetric, so that each dimension of the $n_d-$dimensional space plays the same role, this would require that $\mathbf{W}_{ij}$ is the identity, and the $n_d$ dimensions would thus evolve independently, i.e. the system would be equivalent to $n_d$ identical networks.
% Note that the condition $\mathbf{R}(C)=\mathbf{I}$ supposes that any vector defined on node $v_i$ remains unchanged along a cycle.

One way to relax this condition would be to impose that only vectors along certain directions remain unchanged.
% Take the case when $\mathbf{R}(e_k)$ is the same for each edge and that they have an eigenvalue equal to $1$ - a natural condition to ensure the conservation of a quantity along an edge - with eigenvector $\mathbf{x}$.
We can consider the case when the matrix weights of all edges in $C$ share the same eigenvector $\mathbf{x}$ associated with eigenvalue $1$ - a natural condition to ensure the conservation of a quantity along an edge.
Then, coherence would be verified within this subspace spanned by $\mathbf{x}$, making it possible for non-trivial solutions for dynamical processes defined on MWNs.
% Alternative formulations could also consider cases when the matrices on each edge can be different but commute with each other, and thus have the same eigenvectors. We leave this question for future research.
This definition of coherence could even be sufficient in the long time limit, in situations when $1$ is the largest eigenvalue - this is the case for the transition matrix of random walks for instance -, and its corresponding eigenvector asymptotically dominates the dynamics.
Alternative formulations could also consider cases when the shared eigenvector is not always associated with eigenvalue $1$ but the product of such eigenvalues is $1$ along $C$.
We leave these questions for future research.

\subsection{Additional mathematical results.}
In this section, we give formal statements of the theoretical results presented in the paper and their detailed proofs.

\begin{proposition}
    The supra-Laplacian $\mathcal{L}$ is positive semi-definite.
    \label{pro:L-semi-definite}
\end{proposition}
\begin{proof}[Proof of Proposition \ref{pro:L-semi-definite}]
    We write the supra-Laplacian matrix as a sum over the edges of $G$
    \begin{align*}
        \mathcal{L} = \sum_{(v_i,v_j), (v_j,v_i)\in E}\mathcal{L}^{\{i,j\}},
    \end{align*}
    where $\mathcal{L}^{\{i,j\}}\in\mathbb{R}^{nn_d\times nn_d}$ also has block structure, and if we denote its $k,l$ block as $\mathbf{L}^{\{i,j\}}_{kl}\in\mathbb{R}^{n_d\times n_d}$, the only non-zero blocks are
    \begin{align*}
        \mathbf{L}^{\{i,j\}}_{ii} &= \mathbf{L}^{\{i,j\}}_{jj} = \norm{W_{ij}}_2\mathbf{I} = w_{ij}\mathbf{I}\\
        \mathbf{L}^{\{i,j\}}_{ij} &= \mathbf{L}^{\{i,j\}T}_{ji} = -\mathbf{W}_{ij} = -w_{ij}\mathbf{R}_{ij}.
    \end{align*}
    For all $\mathbf{x} = (\mathbf{x}_1; \mathbf{x}_2; \dots; \mathbf{x}_n)$ with $\mathbf{x}_i\in\mathbb{R}^{n_d}$ for each $i$,
    \begin{align*}
        \mathbf{x}^T&\mathcal{L}^{\{i,j\}}\mathbf{x} \\
        &= \mathbf{x}_i^Tw_{ij}\mathbf{x}_i + \mathbf{x}_j^Tw_{ij}\mathbf{x}_j - \mathbf{x}_i^Tw_{ij}\mathbf{R}_{ij}\mathbf{x}_j - \mathbf{x}_j^Tw_{ij}\mathbf{R}_{ij}^T\mathbf{x}_i\\
        &= w_{ij}(\mathbf{x}_i^T\mathbf{x}_i + \mathbf{x}_j^T\mathbf{x}_j - \mathbf{x}_i^T\mathbf{R}_{ij}\mathbf{x}_j - \mathbf{x}_j^T\mathbf{R}_{ij}^T\mathbf{x}_i)\\
        &= w_{ij}(\mathbf{x}_i - \mathbf{R}_{ij}\mathbf{x}_j)^T(\mathbf{x}_i - \mathbf{R}_{ij}\mathbf{x}_j)\\
        &= w_{ij}\norm{\mathbf{x}_i - \mathbf{R}_{ij}\mathbf{x}_j}_2^2 \ge 0.
    \end{align*}
    Hence,
    \begin{align*}
        \mathbf{x}^T\mathcal{L}\mathbf{x} = \sum_{(v_i,v_j), (v_j,v_i)\in E}w_{ij}\norm{\mathbf{x}_i - \mathbf{R}_{ij}\mathbf{x}_j}_2^2 \ge 0.
        % \label{equ:laplacian-bilinear}
    \end{align*}
\end{proof}

\begin{definition}
    Let $G = (V, E, \mathcal{W})$ be a matrix-weighted graph.
    $G$ is \textit{coherent} if and only if the transformation of every directed cycle is $\mathbf{I}$.
    \label{def:balance}
\end{definition}

\begin{theorem}[structural theorem for balance]
    A matrix-weighted graph $G$ is coherent if and only if there is a partition $\{V_i\}_{i=1}^{l_p}$ s.t. (i) any edges within each node subset have transformation $\mathbf{I}$, (ii) any edges between the same pair of node subsets have the same transformation, and (iii) if we consider each node subset as a super node, then the transformation of any directed cycle is $\mathbf{I}$.
    \label{the:balance-part}
\end{theorem}
\begin{proof}[Proof of Theorem \ref{the:balance-part}]
    Suppose that such a partition exists. Then there are two possible cases for each directed cycle: it either completely lies in one of the node subset, thus has transformation $\mathbf{I}$, or contains one or more directed cycles by considering each node subset as a super node and edges within each node subset, thus the overall transformation is still $\mathbf{I}$. Hence, the graph is coherent by Definition \ref{def:balance}.

    Now suppose the graph is coherent. We first note that all matrix weights are then orthogonal matrices, i.e., $\mathbf{R}_{ij}^T = \mathbf{R}_{ij}^{-1}$ for all $(v_i, v_j)\in E$. We then construct the partition by applying the following rule to each node $v_i$: we group its neighbors together in the same node subset if the corresponding edges from $v_i$ have the same transformation matrix, and we also group $v_i$ to the node subset corresponding to transformation $\mathbf{I}$ if there is any, and as a separate one otherwise (the ``grouping rule'') hereafter). We record the nodes to which the group rule has been applied after step $k$ as $\mathcal{V}_k$ and the resulting set of node subsets as $\mathcal{P}_k$, thus $\mathcal{V}_0 = \emptyset$ and $\mathcal{P}_0 = \emptyset$. We will show by induction that the group rule can be applied to all nodes without causing any conflict.

    (i) When $k=1$, it is trivial that applying the group rule does not cause conflict since $\mathcal{V}_0 = \emptyset$ and $\mathcal{P}_0 = \emptyset$.

    (ii) Suppose that applying the grouping rule does not cause any conflict up to step $k-1$. At step $k$, we select $v_k$ as a neighbor of nodes in $\mathcal{V}_{k-1}$, and suppose by contradiction that applying the grouping rule to $v_k$ will cause conflicts. Then there are four possibilities:
    \begin{enumerate}[label={\arabic*)}]
        \item $\exists v_{l}\in \Gamma(v_k)$ s.t. $v_l, v_k$ are in the same node subset in $\mathcal{P}_{k-1}$ but $\mathbf{R}_{kl} \ne \mathbf{I}$;
        \item $\exists v_{l}\in \Gamma(v_k)$ s.t. $v_l, v_k$ are in a different node subsets in $\mathcal{P}_{k-1}$ but $\mathbf{R}_{kl} = \mathbf{I}$;
        \item $\exists v_l,v_m\in \Gamma(v_k)$ s.t. $v_l, v_m$ are in the same node subset in $\mathcal{P}_{k-1}$, but $\mathbf{R}_{kl} \ne \mathbf{R}_{km}$;
        \item $\exists v_l,v_m\in \Gamma(v_k)$ s.t. $v_l, v_m$ are in the different node subsets in $\mathcal{P}_{k-1}$, but $\mathbf{R}_{kl} = \mathbf{R}_{km}$,
    \end{enumerate}
    where $\Gamma(v_i) = \{v_j: (v_i,v_j)\in E\}$. In case 1), $v_l\notin \mathcal{V}_{k-1}$ since $\mathbf{R}_{kl} \ne \mathbf{I}$ and there is no conflict in $\mathcal{P}_{k-1}$. Then $\exists v_h\in \mathcal{V}_{k-1}$ s.t. $v_k, v_l\in\Gamma(v_h)$ and $\mathbf{R}_{hl} = \mathbf{R}_{hk}$. Hence, $\mathbf{R}_{hk}\mathbf{R}_{kl}\mathbf{R}_{lh} \ne \mathbf{I}$, which contradicts the graph being coherent. The contradiction in case 2) can be shown in a similar manner.

    In case 3), since there is no conflict in $\mathcal{P}_{k-1}$, then either $(v_l,v_m)\in E$ with $\mathbf{R}_{lm} = \mathbf{I}$, or $\exists v_h\in \mathcal{V}_{k-1}$ s.t. $v_l, v_m\in \Gamma(v_h)$ and $\mathbf{R}_{hl} = \mathbf{R}_{hm}$. Then in the former, $\mathbf{R}_{kl}\mathbf{R}_{lm}\mathbf{R}_{mk} \ne \mathbf{I}$, while in the latter, $\mathbf{R}_{lh}\mathbf{R}_{hm}\mathbf{R}_{mk}\mathbf{R}_{kl} \ne \mathbf{I}$, both of which contradict the graph being coherent. The contradiction in case 4) can be shown in a similar manner. Hence, applying the grouping rule to $v_k$ will not cause any conflict.

    (iii) Finally, it is straightforward to check that the resulting partition has the features as stated.
\end{proof}

\begin{theorem}
    The supra-Laplacian matrix $\mathcal{L}$ has eigenvalue $0$ if and only if $G$ is coherent. Furthermore, if $\mathcal{L}$ has eigenvalue $0$, the corresponding eigenvectors span the following subspace
    \begin{align*}
        \mathbf{X} = \mathcal{S}^T(\mathbf{1}\otimes \mathbf{I}),
    \end{align*}
    where $\mathcal{S}$ is a block diagonal matrix that encodes the decomposition in Theorem \ref{the:balance-part}. Specifically, let $\{V_i\}_{i=1}^{l_p}$ denote the corresponding partition, and $\mathcal{S}$ has the $i$-th block on the diagonal being $\mathbf{S}_{1\sigma(i)}$, where $\sigma(\cdot)$ returns the node subset that a node is associated with, and $\mathbf{S}_{hl}$ is the transformation of a directed path from nodes in $V_h$ to $V_l$.
    \label{the:balance-L}
\end{theorem}
\begin{proof}[Proof of Theorem \ref{the:balance-L}]
    $\mathcal{L}$ has eigenvalue $0$ if and only if there exists $\mathbf{x} = (\mathbf{x}_1; \mathbf{x}_2; \dots; \mathbf{x}_n)$ with $\mathbf{x}_i\in\mathbb{R}^{n_d}$ for each $i$, s.t.
    \begin{align}
        \mathbf{x}^T\mathcal{L}\mathbf{x}
        &= \sum_{(i,j), (j,i)\in E}w_{ij}\norm{\mathbf{x}_i - \mathbf{R}_{ij}\mathbf{x}_j}_2^2 = 0 \nonumber\\
        \Leftrightarrow \mathbf{x}_i &= \mathbf{R}_{ij}\mathbf{x}_j, \quad \forall (i,j)\in E.
        \label{equ:balance-L-cond1}
    \end{align}
    Condition \eqref{equ:balance-L-cond1} is true if and only if for any directed cycle in $G$, $C = \bigcup_{j=1}^{l}(v_{i_j}, v_{i_{j+1}})$ with the convention that $v_{i_{l+1}} = v_{i_1}$, we have
    \begin{align*}
        \mathbf{R}_{i_1i_2}\mathbf{R}_{i_2i_3}\dots \mathbf{R}_{i_li_{l+1}} = \mathbf{I},
    \end{align*}
    i.e., $G$ is coherent.

    In the case that $0$ is an eigenvalue of $\mathcal{L}$, or $G$ is coherent, from Theorem \ref{the:balance-part}, we can write the supra-Laplacian matrix as
    \begin{align*}
        \mathcal{S}^T\bar{\mathcal{L}}\mathcal{S},
    \end{align*}
    where $\mathcal{S}$ is as defined in Theorem \ref{the:balance-L} and $\bar{\mathcal{L}}$ is the supra-Laplacian with identity matrix $\mathbf{I}\in \mathbb{R}^{n_d\times n_d}$ as the transformation of each edge,
    \begin{align*}
        \bar{\mathcal{L}} = \bar{\mathbf{L}}\otimes \mathbf{I},
    \end{align*}
    where
    \begin{align*}
        \bar{\mathbf{L}} = \mathbf{D} - \bar{\mathbf{W}},
    \end{align*}
    and $\bar{\mathbf{W}} = (w_{ij})\in \mathbb{R}^{n\times n}$ only contains the magnitude information. We note that $\bar{\mathbf{L}}$ is the (ordinary) graph Laplacian corresponding to an undirected graph $\bar{G} = (V, \bar{E}, \bar{\mathbf{W}})$, thus $0$ is an eigenvalue of $\bar{\mathbf{L}}$. Since $\bar{G}$ is connected, there is only one eigenvector associated with $0$, the all-one vector $\mathbf{1}$. The eigenvectors of $\mathcal{L}$ associated with eigenvalue $0$ can be obtained accordingly.
\end{proof}

\begin{proposition}
    The supra-transition matrix $\mathcal{P}$ has spectral radius no greater than $1$, i.e., $\rho(\mathcal{P}) \le 1$ where $\rho(\mathcal{P}) = \max\{\abs{\lambda}: \lambda \text{ eigenvalue of } \mathcal{P}\}$.
    Furthermore, $\mathcal{P}$ has eigenvalue $1$ if and only if $G$ is coherent, and has eigenvalue $-1$ if and only if $G_n=(V, E, -\mathcal{W})$, the graph obtained by negating each transformation matrix of each edge, is coherent.
    \label{pro:balance-P-rho}
\end{proposition}
\begin{proof}
    We know that $\mathcal{P} = \mathcal{I} - \mathcal{L}_{rw}$, and $\mathcal{L}_{rw}$ is positive semi-definite from Proposition \ref{pro:L-semi-definite}. Hence,
    \begin{align*}
        \lambda_{max}(\mathcal{P}) \le 1,
    \end{align*}
    where $\lambda_{max}(\cdot)$ returns the maximum eigenvalue of a matrix. Meanwhile, since $-\mathcal{P}$ is the supra-transition matrix of $G_n$,
    \begin{align*}
        \lambda_{max}(-\mathcal{P}) \le 1 \quad \Leftrightarrow \quad \lambda_{min}(\mathcal{P}) \ge -1,
    \end{align*}
    where $\lambda_{min}(\cdot)$ returns the minimum eigenvalue of a matrix. Hence, we have
    \begin{align*}
        \rho(\mathcal{P}) \le 1.
    \end{align*}
    Specifically, $\mathcal{P}$ has eigenvalue $1$ if and only if $0$ is an eigenvalue of $\mathcal{L}_{rw}$, if and only if $G$ is coherent, from Theorem \ref{the:balance-L}. The case of $-1$ being an eigenvalue of $\mathcal{P}$ can be obtained in a similar way.
\end{proof}

\begin{proposition}
    If $G$ is coherent, then the consensus dynamics will reach the steady state $\mathbf{y}_j^*$ for each node $v_j$:
    \begin{displaymath}
        \mathbf{y}_j^{*} = \mathbf{S}_{\sigma(j)1}\bar{\mathbf{y}}(0)/n
    \end{displaymath}
    where $\bar{\mathbf{y}}(0) = \sum_{i=1}^n\mathbf{S}_{1\sigma(i)}\mathbf{y}_i(0)$, is the weighted average of the initial state vector with its characteristic transformation, $\{\mathbf{y}_i(0)\}$ is the set of initial state vectors, $\mathbf{S}_{hl}$ and $\sigma(\cdot)$ are as defined in Theorem \ref{the:balance-L}.
    \label{pro:consensus-balance-steady}
\end{proposition}
\begin{proof}[Proof of Proposition \ref{pro:consensus-balance-steady}]
    From $\Dot{\mathbf{y}} = -\mathcal{L}\mathbf{y}$, we can solve that
    \begin{align*}
        \mathbf{y}(t) = \exp(-\mathcal{L}t)\mathbf{y}(0)
    \end{align*}
    where $\mathbf{y}(0) = (\mathbf{y}_1(0); \dots ; \mathbf{y}_n(0))$ is the initial state vector. Since $\mathcal{L}$ is symmetric, we can consider its unitary decomposition in terms of its eigenvalues $\lambda_1 \le \dots \le \lambda_{nn_d}$ and the associated orthonormal eigenvectors $\mathbf{x}_1, \dots, \mathbf{x}_{nn_d}$, respectively,
    \begin{align*}
        \mathcal{L} = \sum_{i=1}^{nn_d}\lambda_i\mathbf{x}_i\mathbf{x}_i^T.
    \end{align*}
    Hence, we can write
    \begin{align*}
        \mathbf{y}(t) = \sum_{i=1}^{nn_d}\exp(-\lambda_i t)\mathbf{x}_i\mathbf{x}_i^T\mathbf{y}(0).
    \end{align*}
    Then in the case that $G$ is coherent,
    \begin{align*}
        \lim_{t\to \infty}\mathbf{y}(t)
        &= \sum_{i=1}^{n_d}\mathbf{x}_i\mathbf{x}_i^T\mathbf{y}(0)\\
        &= \left(\mathcal{S}^T(\mathbf{1}\otimes \mathbf{I})/\sqrt{n}\right)\left(\mathcal{S}^T(\mathbf{1}\otimes \mathbf{I})/\sqrt{n}\right)^T\mathbf{y}(0),
    \end{align*}
    where $\mathbf{x}_1, \dots, \mathbf{x}_{n_d}$ are the eigenvector associated with eigenvalue $\lambda_1 = \dots = \lambda_{n_d} = 0$, and the second equality is obtained from Theorem \ref{the:balance-L} and appropriate scaling for the eigenvectors to be orthonormal. Hence, the steady state $\mathbf{y}_j^*$ for each node $v_j$ can be obtained accordingly.
\end{proof}

\begin{proposition}
    If $G$ is coherent and is not bipartite, then the random walk will reach the steady state $\mathbf{y}_j^*$ for each node $v_j$:
    \begin{displaymath}
        \mathbf{y}_j^{*T} = \bar{\mathbf{y}}(0)^T\mathbf{S}_{1\sigma(j)}d_j/(2m)
    \end{displaymath}
    where $\bar{\mathbf{y}}(0)^T = \sum_{i=1}^n\mathbf{y}_i(0)^T\mathbf{S}_{\sigma(i)1}$, is the weighted average of the initial characteristic vector with its characteristic transformation, $\{\mathbf{y}_i(0)\}$ is the set of initial characteristic vectors, $2m = \sum_{j}d_j$, $\mathbf{S}_{hl}$ and $\sigma(\cdot)$ are as defined in Theorem \ref{the:balance-L}.
    \label{pro:rw-steady}
\end{proposition}
\begin{proof}[Proof of Proposition \ref{pro:rw-steady}]
    For random walks,
    \begin{align*}
        \mathbf{y}(t)^T = \mathbf{y}(0)^T\mathcal{P}^{t}.
    \end{align*}
    Note that $\mathcal{P}$ is not symmetric, but is similar to a symmetric matrix $\mathcal{P}_s = \mathcal{D}^{-1/2}\mathcal{W}\mathcal{D}^{-1/2}$, where
    $\mathcal{P} = \mathcal{D}^{-1/2}\mathcal{P}_s\mathcal{D}^{1/2}$ thus for each eigenpair $(\lambda, \mathcal{D}^{1/2}\mathbf{x})$ of $\mathcal{P}_s$,  $(\lambda, \mathbf{x})$ is a right eigenpair of $\mathcal{P}$. We denote the eigenvalues by $\lambda_1 \le \dots, \le \lambda_{nn_d}$, with the associated eigenvectors of $\mathcal{P}$ by $\mathbf{x}_1, \dots, \mathbf{x}_{nn_d}$, thus the eigenvectors of $\mathcal{P}_s$ are $\mathcal{D}^{1/2}\mathbf{x}_1, \dots, \mathcal{D}^{1/2}\mathbf{x}_{nn_d}$. For illustrative purposes, we assume $\mathcal{D}^{1/2}\mathbf{x}_1, \dots, \mathcal{D}^{1/2}\mathbf{x}_{nn_d}$ to be orthonormal.

    We know that $\mathcal{P}_s = \mathcal{I} - \mathcal{L}_{sys}$, where $\mathcal{L}_{sys} = \mathcal{D}^{-1/2}\mathcal{L}\mathcal{D}^{-1/2}$.
    Hence, $\mathcal{L}_{sys}$ and $\mathcal{P}_{s}$ share the same eigenvectors. Meanwhile, $\mathbf{x}$ is an eigenvector of $\mathcal{L}$ associated with eigenvalue $0$ if and only if $\mathcal{D}^{1/2}\mathbf{x}$ is an eigenvector of $\mathcal{L}_{sys}$ associated with $0$. Hence, the eigenvectors of $\mathcal{P}_s$ associated with $1$ span the subspace,
    \begin{align*}
        \mathbf{X} = \mathcal{D}^{1/2}\mathcal{S}^T(\mathbf{1}\otimes \mathbf{I}),
    \end{align*}
    where $\mathcal{S}$ is a block diagonal matrix that encodes the decomposition in Theorem \ref{the:balance-part}, from Theorem \ref{the:balance-L}.

    In the case that $G$ is coherent and is not bipartite, $1$ is an eigenvalue of $\mathcal{P}$ but not $-1$, from Proposition \ref{pro:balance-P-rho}. Hence, we have
    \begin{align*}
        \lim_{t\to\infty} \mathcal{P}_s^t &= \lim_{t\to\infty}\sum_{i=1}^{nn_d}\lambda_i^t\left(\mathcal{D}^{1/2}\mathbf{x}_i\right)\left(\mathcal{D}^{1/2}\mathbf{x}_i\right)^T\\
        &= \left(\mathcal{D}^{1/2}\mathcal{S}^T(\mathbf{1}\otimes \mathbf{I})/\sqrt{2m}\right)\left(\mathcal{D}^{1/2}\mathcal{S}^T(\mathbf{1}\otimes \mathbf{I})\sqrt{2m}\right)^T,
    \end{align*}
    where $2m = \sum_{j}d_j$. Hence,
    \begin{align*}
       \lim_{t\to\infty}\mathbf{y}(t)^T &= \lim_{t\to\infty}\mathbf{y}(0)^T\mathcal{P}^{t} = \lim_{t\to\infty}\mathbf{y}(0)^T\mathcal{D}^{-1/2}\mathcal{P}_s^t\mathcal{D}^{1/2} \\
       % &= \mathbf{y}(0)^T\mathcal{D}^{-1/2}\left(\mathcal{D}^{1/2}\mathcal{S}^T(\mathbf{1}\otimes \mathbf{I})/\sqrt{2m}\right)\left(\mathcal{D}^{1/2}\mathcal{S}^T(\mathbf{1}\otimes \mathbf{I})\sqrt{2m}\right)^T\mathcal{D}^{1/2} \\
       &= \mathbf{y}(0)^T\mathcal{S}^T(\mathbf{1}\otimes \mathbf{I})/\sqrt{2m}\left(\mathcal{S}^T(\mathbf{1}\otimes \mathbf{I})\sqrt{2m}\right)^T\mathcal{D}.
    \end{align*}
    The steady state $\mathbf{y}_j^*$ for each node $v_j$ can be obtained accordingly.
\end{proof}

\subsection{Robustness Analysis}
\label{sec:sim-results-robustness}

We observe that our results are consistently robust across different numbers of communities and dimensions of the state space (Figs.~\ref{fig:sim-results-consensus-dynamics-var-dim}, \ref{fig:sim-results-consensus-dynamics-var-com}, \ref{fig:sim-results-random-walk-var-dim}, and \ref{fig:sim-results-random-walk-var-com}).
The results demonstrate that the consensus dynamics for a three-block MSBM with 120 nodes in a 3-dimensional state space, and for a two-community MSBM with 120 nodes in a 3-dimensional state space align with our previous observations for three-block models in 10-dimensional space (Fig.~\ref{fig:sim-results-consensus-dynamics-var-dim} and \ref{fig:sim-results-consensus-dynamics-var-com}). In all cases, we see similar patterns of convergence to steady states, with coherent cases showing convergence to the theoretical steady state and incoherent cases exhibiting initial convergence to the theoretical steady states followed by stable deviations.
Similarly, for the random walk dynamics, the number of communities and the dimension of the state space do not affect the convergence patterns (Figs.~\ref{fig:sim-results-random-walk-var-dim} and \ref{fig:sim-results-random-walk-var-com}).
The impact of topological community structure on the dynamics is also consistent across these different configurations. This consistency across varying numbers of communities and state space dimensions underscores the robustness of our findings.

\begin{figure*}[htbp]
    \centering
    \includegraphics[width=\linewidth]{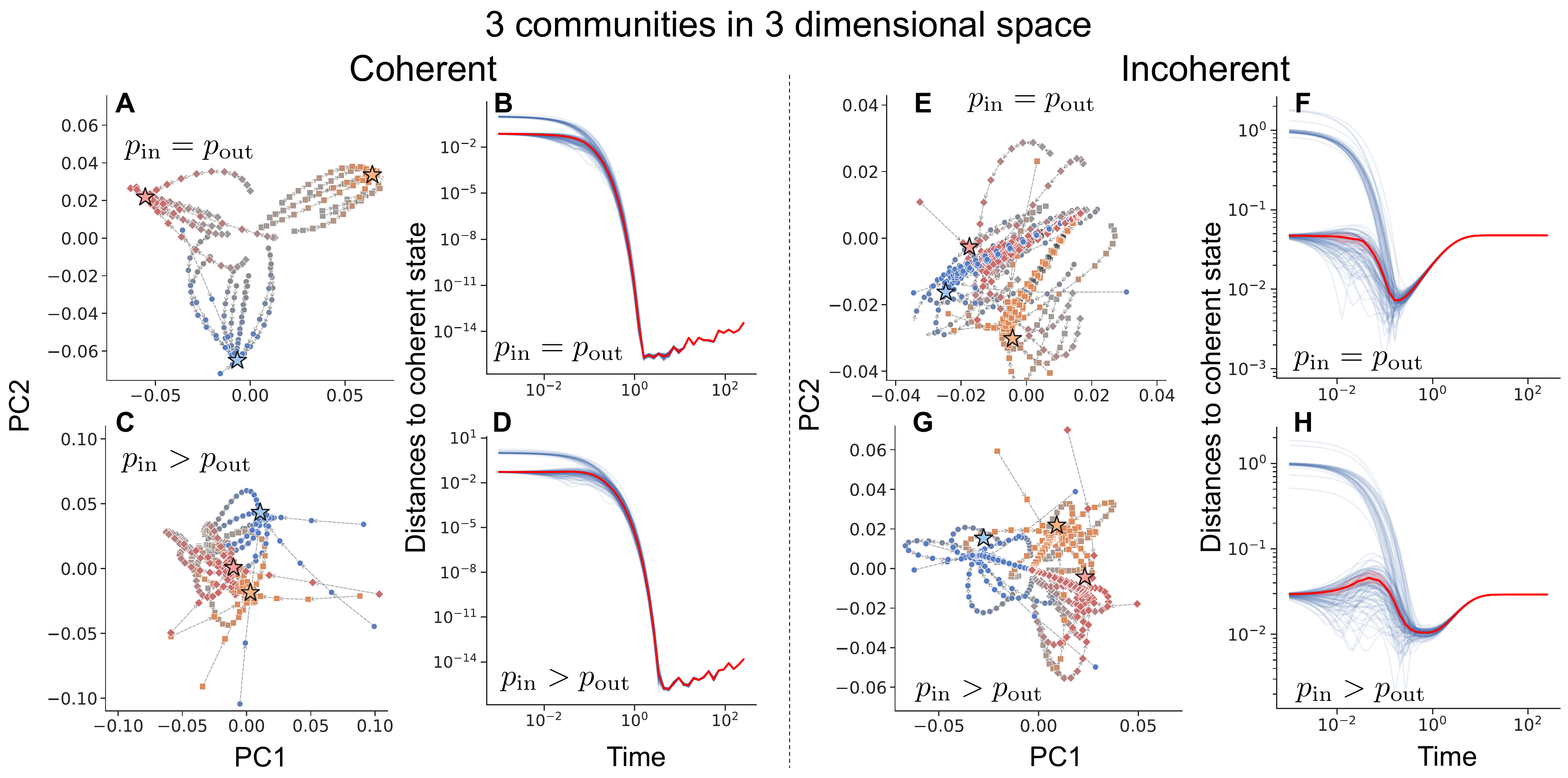}
    \caption{Consensus dynamics in a three-block MSBM with 120 nodes and 3-dimensional state space.
    Red, yellow, and blue represent nodes in blocks 1, 2, and 3, respectively.
    The stars indicate the theoretical steady states.
    We set $n_d = 3$, $p_{\text{in}} = 0.3$, $p_{\text{out}} = \{0.1, 0.3\}$.
    {\bf A--C}: Coherent case, no topological community structure.
    {\bf D--F}: Coherent case, strong topological community structure.
    {\bf G--I}: Incoherent case, no topological community structure.
    {\bf J--L}: Incoherent case, strong topological community structure.
    Panels A, D, G, J show PCA projections of state vector trajectories.
    Panels B, E, H, K show distances to steady state (blue: individual nodes, red: average).
    Panels C, F, I, L show distances to origin (blue: individual nodes, red: average).
    The colorband indicates the 95\% confidence interval estimated by $10^3$ bootstrap sampling.
    }
    \label{fig:sim-results-consensus-dynamics-var-dim}
\end{figure*}

\begin{figure*}[htbp]
    \centering
    \includegraphics[width=\linewidth]{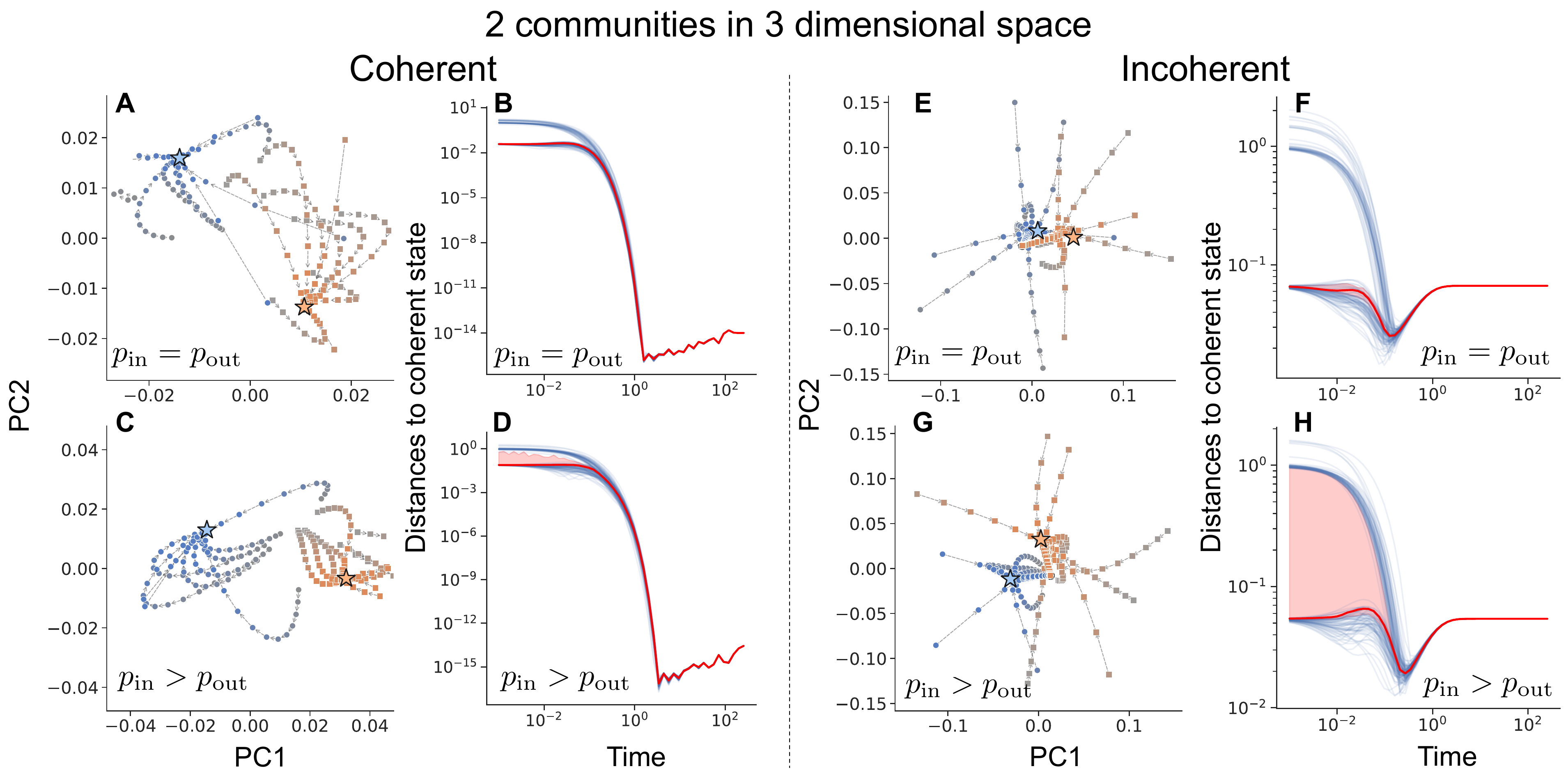}
    \caption{Consensus dynamics in a three-block MSBM with 120 nodes of two communities and 3-dimensional state space.
    Yellow and blue represent nodes in blocks 1 and 2, respectively.
    The stars indicate the theoretical steady states.
    We set $n_d = 3$, $p_{\text{in}} = 0.3$, $p_{\text{out}} = \{0.1, 0.3\}$.
    {\bf A--C}: Coherent case, no topological community structure.
    {\bf D--F}: Coherent case, strong topological community structure.
    {\bf G--I}: Incoherent case, no topological community structure.
    {\bf J--L}: Incoherent case, strong topological community structure.
    Panels A, D, G, J show PCA projections of state vector trajectories.
    Panels B, E, H, K show distances to steady state (blue: individual nodes, red: average).
    Panels C, F, I, L show distances to origin (blue: individual nodes, red: average).
    The colorband indicates the 95\% confidence interval estimated by $10^3$ bootstrap sampling.
    }
    \label{fig:sim-results-consensus-dynamics-var-com}
\end{figure*}

\begin{figure*}[htbp]
    \centering
    \includegraphics[width=\linewidth]{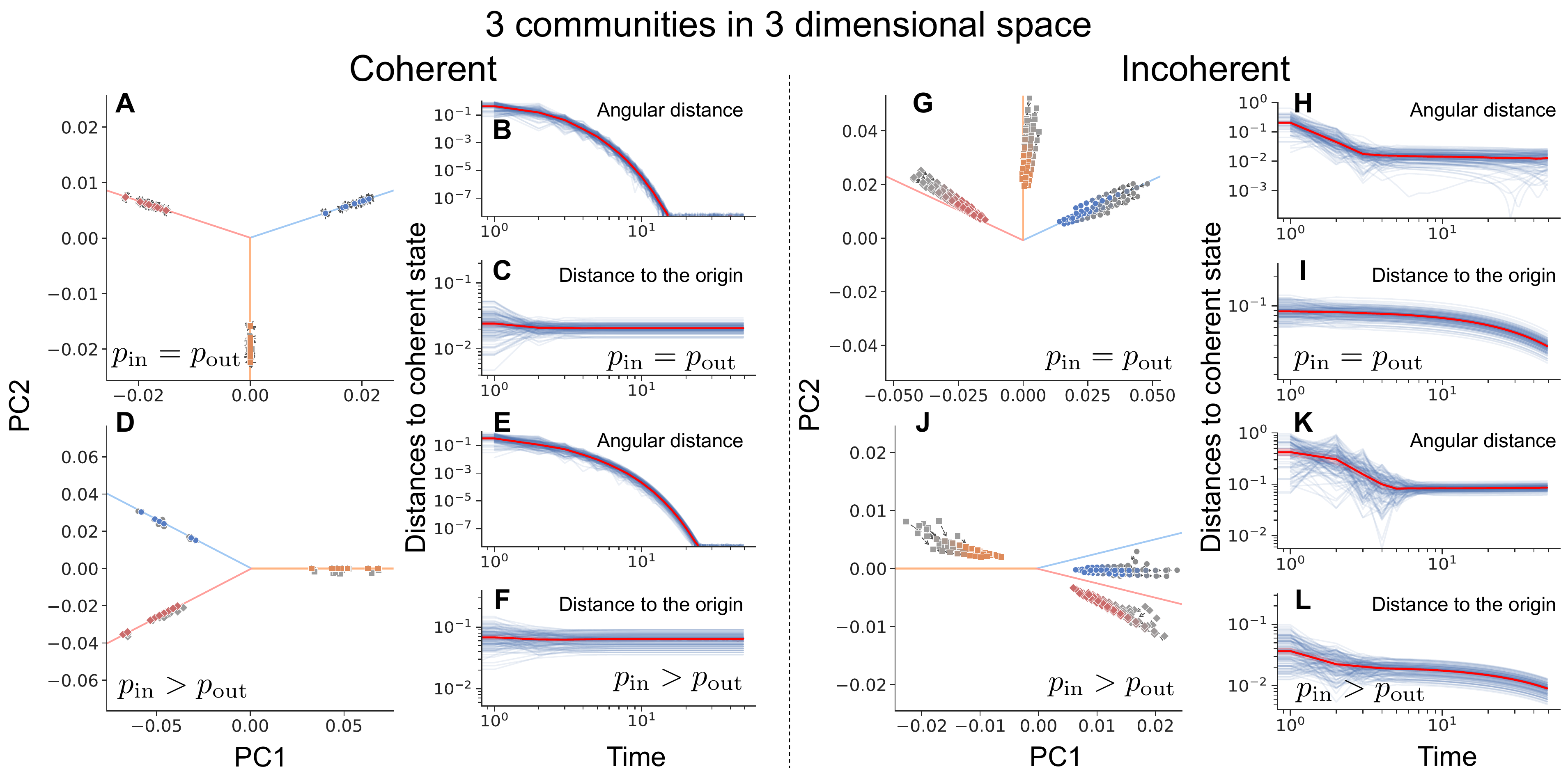}
    \caption{Random walk dynamics in a three-block MSBM with 120 nodes and 3-dimensional state space.
    Red, yellow, and blue represent nodes in blocks 1, 2, and 3, respectively.
    The directions of the theoretical steady states are indicated by solid lines.
    We set $n_d = 3$, $p_{\text{in}} = 0.3$, $p_{\text{out}} = \{0.1, 0.3\}$.
    {\bf A--C}: Coherent case, no topological community structure.
    {\bf D--F}: Coherent case, strong topological community structure.
    {\bf G--I}: Incoherent case, no topological community structure.
    {\bf J--L}: Incoherent case, strong topological community structure.
    Panels A, D, G, J show PCA projections of state vector trajectories.
    Panels B, E, H, K show angular distances to steady state (blue: individual nodes, red: average).
    Panels C, F, I, L show distances to origin (blue: individual nodes, red: average).
    The colorband indicates the 95\% confidence interval estimated by $10^3$ bootstrap sampling.
    }
    \label{fig:sim-results-random-walk-var-dim}
\end{figure*}

\begin{figure*}[htbp]
    \centering
    \includegraphics[width=\linewidth]{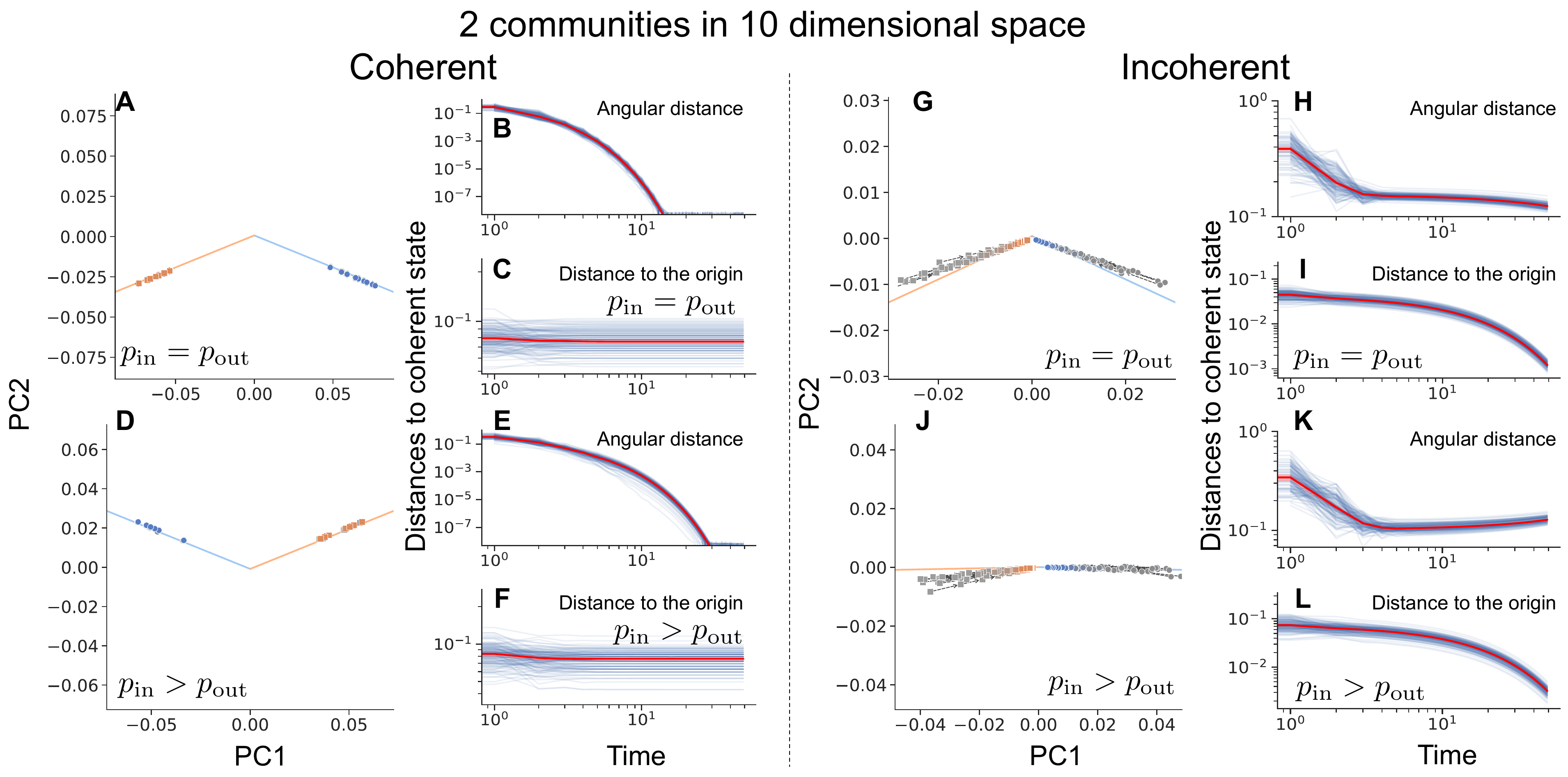}
    \caption{Random walk dynamics in a three-block MSBM with 120 nodes of two communities and 10-dimensional state space.
    Yellow, and blue represent nodes in blocks 1 and 2, respectively.
    The directions of the theoretical steady states are indicated by solid lines.
    We set $n_d = 10$, $p_{\text{in}} = 0.3$, $p_{\text{out}} = \{0.1, 0.3\}$.
    {\bf A--C}: Coherent case, no topological community structure.
    {\bf D--F}: Coherent case, strong topological community structure.
    {\bf G--I}: Incoherent case, no topological community structure.
    {\bf J--L}: Incoherent case, strong topological community structure.
    Panels A, D, G, J show PCA projections of state vector trajectories.
    Panels B, E, H, K show angular distances to steady state (blue: individual nodes, red: average).
    Panels C, F, I, L show distances to origin (blue: individual nodes, red: average).
    The colorband indicates the 95\% confidence interval estimated by $10^3$ bootstrap sampling.
    }
    \label{fig:sim-results-random-walk-var-com}
\end{figure*}

\end{document}